\documentclass[11pt,3p]{elsarticle}

\usepackage{amsfonts}
\usepackage{amsthm}
\usepackage{amsmath}
\usepackage{times}
\usepackage{lineno}
\usepackage{multirow,array}

\newtheorem{theorem}{Theorem}
\newtheorem{proposition}[theorem]{Proposition}
\newtheorem{lemma}[theorem]{Lemma}

\newtheorem{definition}{Definition}

\newtheorem{remark}{Remark}

\newcommand{\F}{\mathcal{F}}
\newcommand{\Su}{\mathcal{S}}
\newcommand{\cO}{\mathcal{O}}

\newcommand{\mina}{\mathcal{M}}

\renewcommand{\epsilon}{\varepsilon}

\newcommand{\wt}{\widetilde}

\begin{document}

\sloppy


\begin{frontmatter}

\title{Minimal Forbidden Factors of Circular Words\tnoteref{note1}}
\tnotetext[note1]{A preliminary version of this paper was presented at the 11th International Conference on Words, WORDS~2017 \cite{FiReRi17}.}

\author[label2]{Gabriele Fici\corref{cor1}}
\ead{gabriele.fici@unipa.it}

\author[label2]{Antonio Restivo}
\ead{antonio.restivo@unipa.it}

\author[label2]{Laura Rizzo}
\ead{rizzolaura88@gmail.com}

\address[label2]{Dipartimento di Matematica e Informatica, Universit\`a di Palermo\\Via Archirafi 34, 90123 Palermo, Italy}

\cortext[cor1]{Corresponding author.}

\journal{Theoretical Computer Science}

\begin{abstract}
Minimal forbidden factors are a useful tool for investigating properties of words and languages. Two factorial languages are distinct if and only if they have different (antifactorial) sets of minimal forbidden factors. There exist  algorithms for computing the minimal forbidden factors of a word, as well as of a regular factorial language. Conversely, Crochemore et al.~[IPL, 1998] gave an algorithm that, given the trie recognizing a finite antifactorial language $M$, computes a DFA recognizing the language whose set of minimal forbidden factors is $M$. In the same paper, they showed that the obtained DFA is minimal if the input trie recognizes the minimal forbidden factors of a single word. We generalize this result to the case of a circular word. We discuss several combinatorial properties of the minimal forbidden factors of a circular word. As a byproduct, we obtain a formal definition of the factor automaton of a circular word. Finally, we investigate the case of minimal forbidden factors of the circular Fibonacci words.
\end{abstract}


\end{frontmatter}

\section{Introduction}

Minimal forbidden factors are a useful combinatorial tool in several areas, ranging from symbolic dynamics to string processing. They have many applications, e.g.~in text compression (where they are also known as \emph{antidictionaries}) \cite{DBLP:conf/icalp/99}, in bioinformatics (where they are also known under the name \emph{minimal absent words}) \cite{Chairungsee2012109,MAW}, etc.
Given a word $w$, a word $v$ is called a \emph{minimal forbidden factor} of $w$ if $v$ does not appear as a factor in $w$ but all the proper factors of $v$ do. For example, over the alphabet $A=\{a,b\}$, the word $w=aabbabb$ has the following minimal forbidden factors: $aaa$, $aba$, $baa$, $babba$, $bbb$.

The theory of minimal forbidden factors is well developed, both from the combinatorial and the algorithmic point of view (see, for instance,~\cite{BeMiReSc00,computingregular,Crochemore98automataand,DBLP:conf/icalp/99,Mignosi02,FICI2006214}). In particular, there exist algorithms for computing the minimal forbidden factors of a single word \cite{Pinho2009,DBLP:conf/isit/FukaeOM12,MAW,PPAM2015}, as well as of a regular factorial language \cite{computingregular}. Conversely, Crochemore et al.~\cite{Crochemore98automataand}, gave an algorithm, called {\sc L-automaton} that, given a trie (tree-like automaton) recognizing a finite antifactorial set $M$, builds a deterministic automaton recognizing the language $L$ whose set of minimal forbidden factors is $M$. The automaton built by the algorithm is not, in general, minimal. However, if $M$ is the set of minimal forbidden factors of a single word $w$, then the algorithm builds the factor automaton of $w$, i.e., the minimal deterministic automaton recognizing the language of factors of $w$ (see~\cite{Crochemore98automataand}). 

The notion of a minimal forbidden factor has been recently extended to the case of circular words (a.k.a.~necklaces) \cite{infcomp,DBLP:conf/isit/OtaM13,6979851}. A circular word can be seen as a sequence of symbols drawn on a circle, where there is no beginning and no end. Although a circular word can be formally defined as an equivalence class of the free monoid under the relation of conjugacy, the fact that in a circular word there is no beginning and no end leads to a less clear definition of the notions like prefix, suffix and factor. For this reason, we consider the set of factors of a circular word $w$ as the (infinite) set of words that appear as a factor in some power of $w$. Although this set is infinite, we show that its set of minimal forbidden factors is always finite, as it coincides with the set of minimal forbidden factors of the word $ww$ that have length bounded by the length of $w$. 

As a main result, we prove that if $M$ is the set of minimal forbidden factors of a circular word, then algorithm {\sc L-automaton} with input a trie recognizing $M$ builds the minimal automaton accepting the set of factors of the circular word. To this end, we use combinatorial properties of the minimal forbidden factors of a circular word. This also allows us to show that it is possible to retrieve a circular word from its set of minimal forbidden factors in linear time with respect to the length of (any linearization of) the word.

Finally, we explore the case of circular Fibonacci words, and give a combinatorial characterization of their minimal forbidden factors.

\section{Preliminaires}

 Let $A$ be a finite alphabet, and let $A^{*}$ be the free monoid generated by $A$ under the operation of concatenation. The elements of $A^*$ are called \emph{words} over $A$. The \emph{length} of a word $w$ is denoted by $|w|$. The \emph{empty word}, denoted by $\epsilon$, is the unique word of length zero and is the neutral element of $A^{*}$.
If $x\in A$ and $w\in A^*$, we let $|w|_x$ denote the number of occurrences of $x$ in $w$. 
 
A \emph{prefix} (resp.~a \emph{suffix}) of a word $w$ is any word $u$ such that $w=uz$ (resp.~$w=zu$) for some word $z$. A \emph{factor} of $w$ is a prefix of a suffix (or, equivalently, a suffix of a prefix) of $w$.  
From the definitions, we have that $\epsilon$ is a prefix, a suffix and a factor of any word. 
A prefix/suffix/factor of a word is \emph{proper} if it is nonempty and does not coincide with the word itself. 
An \emph{occurrence} of a factor $u$ in $w$ is a factorization $w=vuz$. An occurrence of $u$ is \emph{internal} if both $v$ and $z$ are nonempty. 
The set of factors of a word $w$ is denoted by $\F_w$.



The word $\wt{w}$ obtained by reading $w$ from right to left is called the \emph{reversal} (or \emph{mirror image}) of $w$. A \emph{palindrome} is a word $w$ such that $\wt{w}=w$.
In particular, the empty word is  a palindrome. 

The \emph{conjugacy} is the equivalence relation over $A^*$ defined by \[w\sim w' \mbox { if and only if $\exists\ u,v \mid w=uv, w'=vu$}.\]
When the word $w$ is conjugate to the word $w'$, we say that $w$ is a \emph{rotation} of $w'$. 
An equivalence class $[w]$ of the conjugacy relation is called a \emph{circular word} (or \emph{necklace}). A representative of a conjugacy class $[w]$ is called a \emph{linearization} of the circular word $[w]$. Therefore, a circular word $[w]$ can be viewed as the set of all the rotations of a word $w$. 

A word $w$ is \emph{a power} of a word $v$ if there exists a positive integer $k>1$ such that $w=v^k$. 
Conversely, $w$ is {\em primitive} if $w=v^k$ implies $k=1$. Notice that a word is primitive if and only if any of its rotations also is. 
We can therefore extend the definition of primitivity to circular words straightforwardly. Notice that a word $w$ (resp.~a circular word $[w]$) is primitive if and only if there are precisely $|w|$ distinct rotations in the conjugacy class of $w$. 

\begin{remark}\label{rm1}
A circular word can be seen as a word drawn on a circle, where there is no beginning and no end. Therefore, the classical definitions of prefix/suffix/factor of a word lose their meaning for a circular word. In the literature, a factor of a circular word $[w]$ is often defined as a factor of any linearization $w$ of $[w]$. Nevertheless, since there is no beginning and no end, one can define a factor of $w$ as a word that appears as a factor in $w^k$ for some $k$. We will adopt this point of view in this paper.
\end{remark}


\subsection{Minimal Forbidden Factors}

We now recall some basic facts about minimal forbidden factors. For further details and references, the reader may see~\cite{Mignosi02,infcomp}. 

A {\em language} over the alphabet $ A $ is a set of finite words over $ A $, that is, a subset of $A^*$. 
A language is {\em factorial} if it contains all the factors of its words. The \emph{factorial closure} of a language $L$ is the language consisting of all factors of the words in $L$, that is, the language $\F_L=\cup_{w\in L} \F_w$. 

The counterparts of factorial languages are antifactorial languages. A language is called {\em antifactorial} if no word in the language is a proper factor of another word in the language. Dual to the notion of factorial closure, there also exists the notion of \emph{antifactorial part} of a language, obtained by removing the words that are factors of another word in the language.

\begin{definition}
 Given a factorial language $L$ over the alphabet $A$,  the (antifactorial) language of \emph{minimal forbidden factors} of $L$ is defined as \[\mina_L=\{a\in A \mid a\notin L\}\cup\{aub \in A^*\mid a,b\in A,\  aub\notin L,\ au,ub\in L\}.\] 
\end{definition}

A minimal forbidden factor of the language $L$ is therefore a word in $\mina_L$.

Every factorial language $L$ is uniquely determined by its (antifactorial) language of minimal forbidden factors $\mina_L$, through the equation 
\begin{equation}\label{maw}
L= A ^*\setminus  A ^*\mina_L A ^*.
\end{equation}
The converse is also true, since by the definition of a minimal forbidden factor we have \begin{equation}\label{maw2}
\mina_L= A  L\cap L A  \cap ( A ^*\setminus L).
\end{equation}
The previous equations define a bijection between factorial and antifactorial languages.

In the case of a single  word $w$ over an alphabet $A$, the set of minimal forbidden factors of $w$, that we denote by $\mina_w$, is defined as the antifactorial language $\mina_{\F_{w}}$. Indeed, a word $v\in A^*$ is a minimal forbidden factor of the word $w$ if $v$ is a letter of $A$ not appearing in $w$ or $v=aub$, with $a,b\in A$,  $aub\notin \F_w$ and  $au,ub\in \F_w$. 

For example, consider the word $w=aabbabb$ over the alphabet $A=\{a,b\}$. The set of minimal forbidden factors of $w$ is $\mina_w=\{aaa, aba, bbb, baa,babba\}$.  

\begin{remark}
 Applying \eqref{maw} and \eqref{maw2} to the language of factors of a single word, we have that,  given two words $u$ and $v$,  one has $u=v$ if and only if $\mina_u=\mina_v$, i.e., every word can be uniquely represented by its set of minimal forbidden factors.
\end{remark}

An important property of the minimal forbidden factors of a word $w$, which plays a crucial role in algorithmic applications, is that their number is linear in the size of $w$.  Let $w$ be a word of length $n$ over an alphabet $A$ of cardinality $\sigma$. In~\cite{Mignosi02} it is shown that the total number of minimal forbidden factors of $w$ is smaller than or equal to $\sigma n$. Actually, $\mathcal{O}( \sigma  n)$ is a tight asymptotic bound for the number of minimal forbidden factors of $w$ whenever $2 \leq  \sigma  \leq n$ \cite{infcomp}. They can therefore be stored on a trie, whose number of nodes is linear in the size of the word. Recall that a \emph{trie} representing a finite language $L$ is a tree-like deterministic  automaton recognizing $L$, where the set of states is the set of prefixes of words in $L$, the initial state is the empty word $\varepsilon$, the set of final states is a set of \emph{sink} states, and the set of transitions is $\{(u,a,ua) \mid a \in A\}$.

\subsection{Automata for Minimal Forbidden Factors}

Recall that a \emph{deterministic finite automaton} (DFA) is a $5$-tuple $\mathcal{A}=(Q,A,i,T,\delta)$, where $Q$ is the finite set of states, $A$ is the current alphabet, $i$ is the initial state, $T$ the set of terminal (or final) states, and $\delta:(Q\times A)\mapsto Q$ is the transition function. A word is \emph{recognized} (or \emph{accepted}) by $\mathcal{A}$ if reading $w$ character by character from the initial state leads to a final state. The language recognized (or accepted) by $\mathcal{A}$ is the set of all words recognized by $\mathcal{A}$. A language is  \emph{regular} if it is recognized by some DFA. A DFA $\mathcal{A}$ is \emph{minimal} if it has the least number of states among all the DFA's recognizing the same language as $\mathcal{A}$. The minimal DFA is unique.

It follows from basic closure properties of regular languages that the bijection between factorial and antifactorial languages expressed by \eqref{maw} and \eqref{maw2} preserves regularity, i.e., a factorial language is regular if and only if its language of minimal forbidden factors also is.

The \emph{factor automaton} of a word $w$ is the minimal DFA recognizing the (finite) language $\F_w$. The factor automaton of a word of length $n>3$ has at least\footnote{We do not require here and in the remainder of this paper that an automaton be complete. However, to make an automaton complete it is sufficient to add one sink state towards which all missing transitions go.} $n+1$  and at most $2n-2$ states \cite{C86,Bl83}. It can be built in $\mathcal{O}(n)$ time and space by an algorithm that also constructs the \emph{failure function} of the  automaton \cite{AutomataForMatchingPatterns}. The failure function of a state $p$ (different from the initial state) is a link to another state $q$ defined as follows: Let $u$ be a nonempty word and $p=\delta(i,u)$. Then $q=\delta(i,u')$, where $u'$ is the longest suffix of $u$ for which $\delta(i,u)\neq \delta(i,u')$. It can be shown that this definition does not depend on the particular choice of $u$~\cite{Crochemore98automataand}. An example of a factor automaton of a word is displayed in Figure~\ref{fig:exampleFA}.

\begin{figure}[tb]
\begin{center}
\includegraphics[height=33mm]{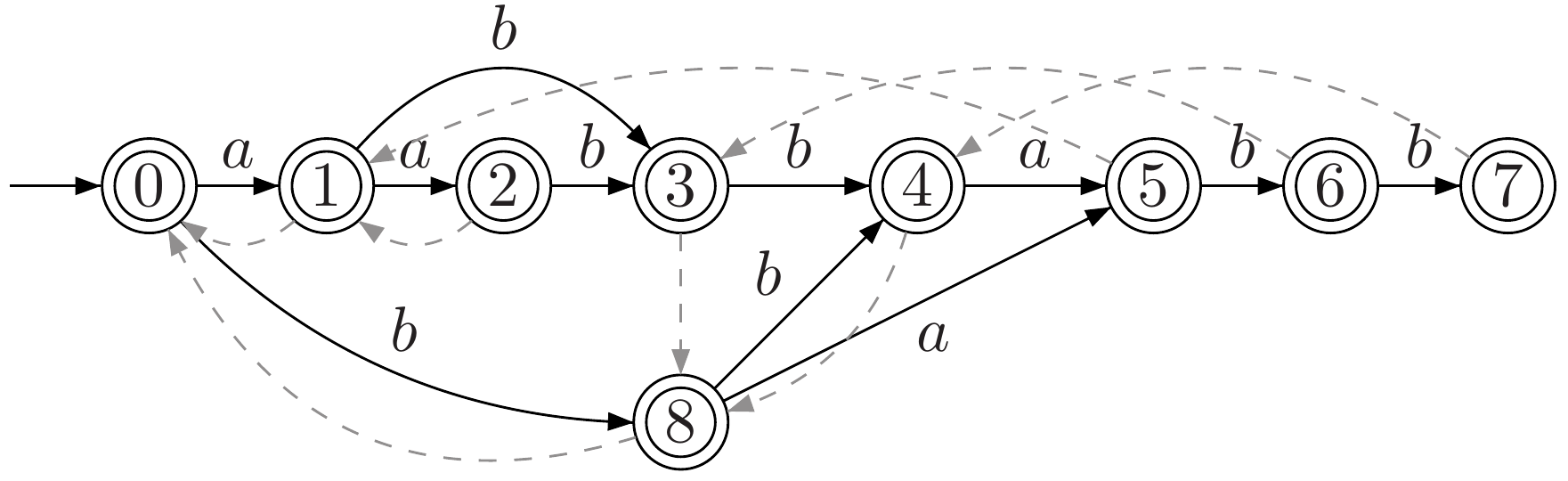}
\end{center}
\caption{The factor automaton of the word $w=aabbabb$. It is the minimal DFA recognizing the language $\F_w$ of factors of $w$. Dashed edges correspond to the failure function links.
\label{fig:exampleFA}}
\end{figure}

In \cite{computingregular}, the authors gave a quadratic-time algorithm to compute the set of minimal
forbidden factors of a regular factorial language $L$. 

However, computing the minimal forbidden factors of a single word can be done in linear time in the length of the word. Algorithm {\sc MF-trie}, presented in 
\cite{Crochemore98automataand}, builds the trie of the set $\mina_w$ having as input the factor automaton of
$w$, together with its failure function. 
Moreover, the states of the output trie recognizing the set $\mina_w$ are the same as those of the factor automaton of $w$, plus some sink states, which are the terminal states with no outgoing edges, corresponding to the minimal forbidden factors. This in particular proves that the size of the trie recognizing the set $\mina_w$ is linear in the length of $w$, since, as already mentioned, the size of the factor automaton of $w$ is linear in the length of $w$. This property is not an immediate consequence of the fact that the \emph{number} of minimal forbidden factors of $w$ is linear in the length of $w$, since in fact the \emph{sum of the lengths} of the minimal forbidden factors of $w$ can be quadratic in the length of $w$ (for example, if $w$ is of the form $ab^na$, then $ab^ia$ is a minimal forbidden factor of $w$ for every $i=0,1,\ldots,n-1$).

Despite these theoretical advantages, algorithm {\sc MF-trie} may not be the best algorithm to be used in applications. In recent years, other algorithms have been introduced  to compute the minimal forbidden factors of a word. The computation of minimal forbidden factors based on the construction of suffix arrays was considered in~\cite{Pinho2009}; although this algorithm has a linear-time performance in practice, 
the worst-case time complexity is $\cO(n^2)$. New $\cO(n)$-time and $\cO(n)$-space suffix-array-based algorithms were presented in~\cite{DBLP:conf/isit/FukaeOM12,MAW,PPAM2015}. 
A more space-efficient solution to compute all minimal forbidden factors in time $\cO(n)$ was also presented in~\cite{Belazzougui2013}.

We have discussed algorithms for computing the set of minimal forbidden factors of a given factorial language. We are now describing an algorithm performing the reverse operation.
Let $M$ be an antifactorial language. We let $L(M)$ denote the (factorial) language avoiding $M$, that is, the language of all the words that do not contain any word of $M$ as a factor. Clearly, from equations \eqref{maw} and \eqref{maw2}, we have that $L(M)$ is the unique language whose  set of minimal forbidden factors is $M$, i.e., the unique language $L$ such that $\mina_{L}=M$.

For a finite antifactorial language $M$, algorithm {\sc L-automaton}  \cite{Crochemore98automataand} builds a DFA recognizing $L(M)$. It is presented in Figure~\ref{fig:laut}. The algorithm runs in linear time in the size of the trie storing the words of $M$. It uses a failure function $f$ defined in a way analogous to the one used for building the factor automaton.

\begin{figure}[tb]
\begin{center}\small
 \fbox{
  \begin{minipage}{10cm}
   \begin{tabbing}
 xxx \= xxx \= xxx \= xxx \kill
{\sc L-automaton} (trie $\mathcal{T}=(Q,A,i,T,\delta')$) \phantom{xxx xxx xxx}\\
\,\,\,1. \> \textbf{for} each $a\in A$\\
\,\,\,2. \> \> \textbf{if} $\delta'(i,a)$ is defined\\
\,\,\,3. \> \> \> $\delta(i,a)\leftarrow \delta'(i,a)$\\
\,\,\,4. \> \> \> $f(\delta(i,a))\leftarrow i$\\
\,\,\,5. \> \> \textbf{else} \\
\,\,\,6. \> \> \> $\delta(i,a)\leftarrow i$\\
\,\,\,7. \> \textbf{for} each state $p\in Q\setminus \{i\}$ in breadth-first search \textbf{and} each $a\in A$\\
\,\,\,8. \> \> \textbf{if} $\delta'(p,a)$ is defined\\
\,\,\,9. \> \> \> $\delta(p,a)\leftarrow \delta'(p,a)$\\
10.\,\,\, \> \> \> $f(\delta(p,a))\leftarrow \delta(f(p),a)$\\
11.\,\,\, \> \> \textbf{else if} $p\notin T$\\
12.\,\,\, \> \> \>   $\delta(p,a)\leftarrow \delta(f(p),a)$\\
13.\,\,\, \> \> \textbf{else} \\
14.\,\,\, \> \> \>  $\delta(p,a)\leftarrow p$\\
15.\,\,\, \> \textbf{return} $(Q,A,i,Q\setminus T,\delta)$
   \end{tabbing}
  \end{minipage}
}
\end{center}
\caption{Algorithm {\sc L-automaton}. It builds an automaton recognizing the language $L(M)$ of words avoiding an antifactorial language $M$ on the input trie $\mathcal{T}$ accepting $M$.\label{fig:laut}}
\end{figure}

The algorithm can be applied for retrieving a word from its set of minimal forbidden factors, and this can be done in linear time in the length of the word, since, as already mentioned, the size of the trie of minimal forbidden factors of a word is linear in the length of the word.
Notice that, even if $M$ is finite, the language $L(M)$ can be finite or infinite. Moreover, even in the case that $L(M)$ is also finite, it can be the language of factors of a single word or of a finite set of words.

Algorithm {\sc L-automaton} builds an automaton recognizing the language $L(M)$ of words avoiding a given antifactorial language $M$, but this automaton is not, in general, minimal. However, the following result holds \cite{Crochemore98automataand}:

\begin{theorem}\label{thm:min}
 If $M$ is the set of the minimal forbidden factors of a finite word $w$, then the automaton output from algorithm  {\sc L-automaton} on the input trie recognizing $M$, after removing sink states, is the factor automaton of $w$, i.e., it is minimal. 
\end{theorem}

To see that the minimality described in the previous theorem does not hold in general, consider for instance the antifactorial language $M=\{aa,ba\}$. It can be easily checked that, taking as input a trie recognizing $M$, algorithm  {\sc L-automaton} outputs an automaton which, after removing sink states, has $3$ states,  while the minimal automaton of the language $L(M)=\{b^n\mid n\geq 0\}\cup\{ab^n\mid n\geq 0\}$ has $2$ states (see Figure~\ref{fig:ex}). 

\begin{figure}[tb]
\centering  
\begin{minipage}{.3\linewidth}
\includegraphics[height=47mm]{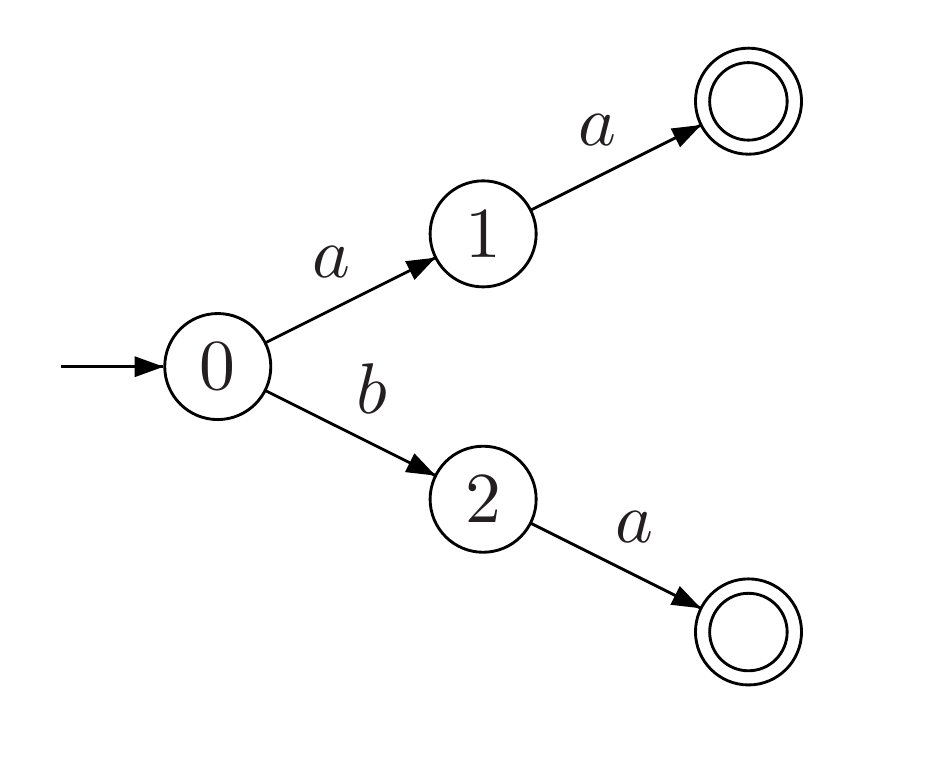}
 \end{minipage}
\begin{minipage}{.3\linewidth}
\includegraphics[height=28mm]{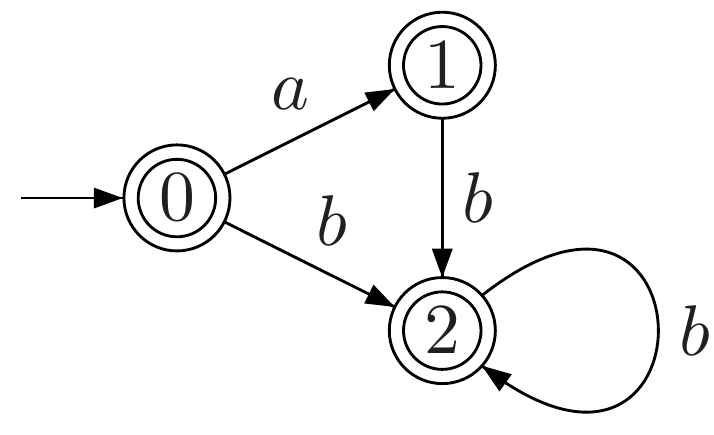}
 \end{minipage}
 \begin{minipage}{.3\linewidth}
\includegraphics[height=13mm]{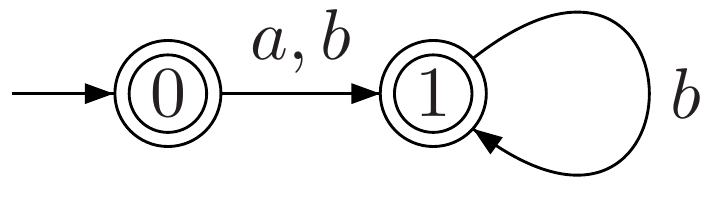}
 \end{minipage}
 \vspace{4mm}
\caption{\label{fig:ex}The trie $\mathcal{T}$ recognizing $M=\{aa,ba\}$ (left), the automaton output from algorithm  {\sc L-automaton} on input $\mathcal{T}$ after removing sink states (center) and the minimal automaton recognizing the language $L(M)=\{b^n\mid n\geq 0\}\cup\{ab^n\mid n\geq 0\}$ (right).}
\end{figure}

We will prove in the next section that this minimality property still holds true in the case of minimal forbidden factors of a circular word. 

From the set $\mina_w$ of the minimal forbidden factors of a
finite word $w$, one can reconstruct the word $w$ in linear time. To this end, one can apply algorithm  {\sc L-automaton} on the trie recognizing $\mina_w$. After deleting
the sink states of the obtained automaton, one can retrieve the longest path starting from the initial
state by using a classical topological sort procedure. This path corresponds to the word $w$.

\section{Minimal Forbidden Factors of a Circular Word}

Given a word $w$, the language \emph{generated} by $w$ is the language $w^*=\{w^k\mid k\geq 0\}=\{\varepsilon, w, ww, www, \ldots\}$. Analogously, the language $L^*$ generated by $L\subset A^*$ is the set of all possible concatenations of words in $L$, i.e.,~$L^* = \{\epsilon\}\cup\{w_1w_2\cdots w_n \mid  w_i \in L \mbox{ for $i = 1, 2,\ldots,n$}\}$.

 Let $w$ be a word of length at least $2$. The language $w^*$ generated by $w$ is not a factorial language, nor is the language generated by all the rotations of $w$. Nevertheless, if we take the factorial closure of the language generated by $w$, then of course we get a factorial language $\F_{w^*}$. Now, if $z$ is conjugate to $w$, then although $w$ and $z$ generate different languages, the factorial closures of the languages they generate coincide, i.e., $\F_{w^*}=\F_{z^*}$. Moreover, for any power $w^k$ of $w$, $k>0$, one clearly has $\F_{w^*}=\F_{(w^k)^*}$.

Based on the previous discussion, and on Remark~\ref{rm1}, we give the following definition: We let the set of \emph{factors of a circular word} $[w]$ be the (factorial) language $\F_{w^*}$,  where $w$ is any linearization of $[w]$. By the previous observation, this definition is independent of the particular choice of the linearization. Moreover, we can suppose that $[w]$ is a primitive circular word.

The set of minimal forbidden factors of the circular word $[w]$ is defined as the set $\mina_{\F_{w^*}}$ of minimal forbidden factors of the language $\F_{w^*}$, where $w$ is any linearization of $[w]$. To simplify the notation, in the remainder of this paper we will let $\mina_{[w]}$ denote the set of minimal forbidden factors of the circular word $[w]$. 

For instance, if $[w]=[aabbabb]$,  then we have
$$\mina_{[w]}=\{aaa, aba, bbb, aabbaa, babbab\}.$$ 

Notice that $\mina_{[w]}$ does not coincide with the set of minimal forbidden factors of the factorial closure of the language of all the rotations of $w$.

Although $\F_{w^*}$ is an infinite language, the set $\mina_{[w]}=\mina_{\F_{w^*}}$ of minimal forbidden factors of $[w]$ is always finite. More precisely, we have the following structural lemma.

\begin{lemma}\label{lem:m}
  \label{lem:circ}
  Let $[w]$ be a circular word and $w$ any linearization of $[w]$. Then 
\begin{equation}\label{eq:cmf}
   \mina_{[w]}=\mina_{ww}\cap A^{\leq |w|}.
\end{equation}
\end{lemma}

\begin{proof}
If $v$ is an element of $\mina_{ww}\cap A^{\leq |w|},$ then clearly  $v\in \mina_{\F_{w^*}}=\mina_{[w]}$.

Conversely, let $aub$, with $a,b\in A$ and $u\in A^*$, be an element in $\mina_{[w]}=\mina_{\F_{w^*}}$ (the case of minimal forbidden factors of length $1$ is straightforward). Then $aub\notin \F_{w^*}$, while $au,ub\in \F_{w^*}$. So, there exists some letter $\bar{b}$ different from $b$ such that $au\bar{b}\in \F_{w^*}$ and a letter $\bar{a}$ different from $a$ such that $\bar{a}ub\in \F_{w^*}$. Therefore, $au,\bar{a}u,ub,u\bar{b}\in \F_{w^*}$.
It is readily verified that any word of length at least $|w|-1$ cannot be extended to the right nor to the left by different letters in $\F_{w^*}$. Hence $|aub|\leq |w|$. Since $au$ and $ub$ are factors of some rotation of $w$, we have $au,ub\in \F_{ww}$, whence $aub\in \mina_{ww}$.
\end{proof}

The equality (\ref{eq:cmf}) was first introduced as a definition for the set of minimal forbidden factors of a circular word in~\cite{6979851}. In fact, it can be efficiently exploited in applications of minimal forbidden factors of circular words \cite{6979851,InfComp18}.

It should be noticed here that the fact that the set $\mina_{[w]}=\mina_{\F_{w^*}}$ is finite although the set $\F_{w^*}$ is not, is a property that cannot be extended to sets of words of cardinality greater than $1$. As an example, let $L=\{b,aa\}$; then the set $\mina_{\F_{L^*}}$ is infinite, as it contains the words $ba^{2n+1}b$ for all $n\geq 1$.

About the number of minimal forbidden factors of a circular word, we have the following bounds.

\begin{lemma}\label{lem:nb}
 Let $[w]$ be a circular word of length $n$ over the alphabet $A$ and let $A(w)$ be the set of letters of $A$ that occur in $w$. Then
\begin{equation}
|A|-1 \leq |\mina_{[w]}| \leq |A|+(n-1)|A(w)|-n.
\end{equation}
In particular, if $|A(w)|=|A|$, then $|\mina_{[w]}|\leq n(|A|-1)$.
\end{lemma}

\begin{proof}
 The inequality $|A|-1 \leq |\mina_{[w]}|$ follows from the fact that for each letter $a\in A$, except at most one, there exists an integer $n_a>0$ such that $a^{n_a}\in \mina_{[w]}$. For the upper bound, we first observe that the minimal forbidden factors of length $1$ of $[w]$ are precisely the elements of $A\setminus A(w)$. We now count the minimal forbidden factors of length greater than $1$. By Lemma~\ref{lem:m}, we know that $\mina_{[w]}=\mina_{ww}\cap A^{\leq |w|}$. Let $ww=w_1w_2\cdots w_{2n}$. Consider a position $i$ in $ww$ such that $n\leq i<2n$. We claim that there are at most $|A|$ distinct elements of $\mina_{[w]}$ of length greater than $1$  whose longest proper prefixes have an occurrence ending in position $i$. Indeed, by contradiction, let $b\in A$ such that there exist $ub,vb\in \mina_{[w]}$ and both $u$ and $v$ occur in $ww$ ending in position $i$. This implies that $ub$ and $vb$ are one suffix of another, against the minimality of the minimal forbidden factors. Since the letter $b$ must be different from the letter of $ww$ occurring in position $i+1$, we therefore have that the number of minimal forbidden factors obtained for $i$ ranging from $n$ to $2n-1$ is at most $n(|A(w)|-1)$. For $i$ such that $1\leq i<n$ (resp.~$i=2n$), if an element $ub\in \mina_{[w]}$, $b\in A$, is such that $u$ has an occurrence in $ww$ ending in position $i$, then $u$ has also an occurrence ending in position $i+n$ (resp.~$n$), so it has already been counted.  Hence,
 $$|\mina_{[w]}| \leq |A|-|A(w)|+n(|A(w)|-1)=|A|+(n-1)|A(w)|-n.$$
\end{proof}

The bounds in the previous lemma are tight. For the lower bound, we have for example that  the set of minimal forbidden factors of the circular word $[w]=[a^n]$ is $A\setminus \{a\}$; for the upper bound, the word $[w]=[a^{n-1}b]$ over the alphabet $A=\{a,b\}$ has $n$ distinct minimal forbidden factors, namely $a^{n}$ and $ba^ib$ for every $i=0,1,\ldots,n-2$. As another example, if $[w]$ is a binary de Bruijn word of order $k$ (and hence length $2^k$), then for every binary word $v$ of length $k$ there exists exactly one letter $a$ such that $va$ is a minimal forbidden factor of $[w]$, hence $|\mina_{[w]}| = |w|$. For instance, if $[w]=[aaababbb]$, then one has $\mina_{[w]}=\{aaaa, aabb, abaa, abba, baab, baba, bbab, bbbb\}$.
Over the alphabet $A=\{a,b,c\}$, for any odd $n$, the circular word $[w]=[a^{\lfloor{n/2\rfloor}}ba^{\lfloor{n/2\rfloor}-1}c]$, of length $n$, has $2n$ distinct minimal forbidden factors. Moreover, over the alphabet $A=\{a_1,a_2,\ldots,a_n\}$, the circular word $[w]=[a_1a_2\cdots a_n]$ has $n(n-1)$ distinct minimal forbidden factors.

We now give a result analogous to Theorem \ref{thm:min} in the case of circular words.

\begin{theorem}\label{thm:min2}
 If $M$ is the set of the minimal forbidden factors of a primitive circular word $[w]$, then the automaton output from algorithm  {\sc L-automaton} on the input trie $\mathcal{T}$ recognizing $M$, after removing sink states, is the minimal automaton recognizing the language $\F_{w^*}$ of factors of $[w]$. 
\end{theorem}

\begin{proof}
Let $\mathcal{A}=(Q,A,i,Q\setminus T,\delta)$ be the automaton output by algorithm {\sc L-automaton} with input the trie $\mathcal{T}$ recognizing the set of the minimal forbidden factors of a circular word $[w]$. Let $w=w_1w_2\cdots w_n$ be a linearization of $[w]$.  By the property of algorithm {\sc L-automaton}, the automaton $\mathcal{A}$ recognizes the language $\F_{w^*}$, since its input is the trie that recognizes the language $\mina_{[w]}=\mina_{\F_{w^*}}$. To prove that $\mathcal{A}$ is minimal, we have to prove that any two states are distinguishable. Suppose by contradiction that there are two nondistinguishable states $p,q\in Q$. By construction, $p$ and $q$ are respectively associated with two proper prefixes, $v_p$ and $v_q$, of words in $\mina_{\F_{w^*}}$, which, by Lemma \ref{lem:circ}, is equal to $\mina_{ww}\cap A^{\leq |w|}$. Therefore, $v_p$ and $v_q$ are factors of $w^*$ of length $\leq |w|$. Hence, they are both factors of $w^2$. Let us then write $w^2=xv_py=x'v_qy'$, with $x$ and $x'$ of minimal length.

Suppose first that there exists $i$ such that $xv_p$ and $x'v_q$ both end in $w_1w_2\cdots w_i$. Then $v_p$ and $v_q$ are one suffix of another. Since $p$ and $q$ are nondistinguishable, there exists a word $z$ such that  $v_pz$ and $v_qz$ end in a sink state, that is, are elements of $\mina_{[w]}$. This is a contradiction since $\mina_{[w]}$ is an antifactorial set and $v_pz$ and $v_qz$ are one suffix of another.

Suppose now that $xv_p$ ends in $w_1w_2\cdots w_i$ and $x'v_q$ ends in $w_1w_2\cdots w_j$ for $i\neq j$. Since $p$ and $q$ are nondistinguishable, for any word $u$ one has that that $v_pu\in \F_{w^*}$ if and only if $v_qu\in \F_{w^*}$. Since $\F_{w^*}$ is a factorial language, we therefore have that there exists a word $z$ of length $|w|$ such that $v_pz$ and $v_qz$ are both in $\F_{w^*}$. But this implies that $z=w_{i+1}w_{i+2}\cdots w_i=w_{j+1}w_{j+2}\cdots w_j$, and this leads to a contradiction since $w$ is primitive and therefore all its rotations are distinct.
\end{proof}

If one is interested in retrieving the circular word $[w]$ from the minimal automaton recognizing  the language $\F_{w^*}$, this can be done with a simple Depth-First-Search procedure in linear time and space with respect to the size of the automaton. Indeed,  the circular word $[w]$ corresponds to a cycle in the (multi-)graph of the automaton, and it can be proved that the size of the minimal automaton  recognizing  the language $\F_{w^*}$ is linear in the length of $w$. This follows from the fact that the size of the output of {\sc L-automaton} is, by construction, the same as the size of its input (except for the sink states), and we have that the size of the trie recognizing the minimal forbidden factors of $[w]$ is linear in the size of $w$ --- as this is a subtrie of  the trie recognizing the minimal forbidden factors of $ww$ (Lemma \ref{lem:m}), and this latter has a size that is linear in the length of $w$, as observed in the previous section.

However, it is possible to give more precise bounds on the size of the minimal automaton  recognizing  the language $\F_{w^*}$ in terms of the length of $w$, as shown below.

\begin{theorem}\label{thm:states}
Let $w$ be a word of length $n$. The minimal automaton  recognizing  the language $\F_{w^*}$ has at most $2n-1$ states.
\end{theorem}

\begin{proof}
Let $w=a_1a_2\cdots a_n$. Let $\Su_{w^*}$ be the set of suffixes of words in the language $w^*$. For every $x\in \F_{w^*}$ different from the empty word, let $E_w(x)=\{i\in\{1,\ldots,n\}\mid x\in \Su_{w^*}a_1a_2\cdots a_i\}$, and let $E_w(\epsilon)=\{1,2,\ldots,n\}$. We define the equivalence $x\equiv_w y$ if and only if $E_w(x)=E_w(y)$. This  equivalence clearly being right-invariant ($x\equiv_w y$ implies $xa\equiv_w ya$ for every letter $a$), there exists a DFA $\mathcal{A}$ recognizing $\F_{w^*}$ whose states are identified with the equivalence classes of $\equiv_w$. It is readily verified that for every $x,y\in \F_{w^*}$, either $E_w(x)$ and $E_w(y)$ are disjoint or one is contained in the other. They therefore form a non-overlapping family of nonempty subsets of $\{1,2,\ldots n\}$, which implies that there are at most $2n-1$ of them (see for example Lemma 1 in \cite{Bl83}, where the value $2n-1$ is replaced by $2n$ because the authors consider the equivalence in $\Sigma^*$ rather than in $\F_{w^*}$ thus allowing the empty set as an extra class corresponding to the additional sink state). This shows that there exists an automaton recognizing $\F_{w^*}$ with at most $2n-1$ states.
\end{proof}

The bound in the previous theorem is tight for $n>2$. As an example, consider $w=ab^{n-1}$. The states of the minimal automaton recognizing  the language $\F_{w^*}$ are identified with the classes of the equivalence relation defined on $\F_{w^*}$ by: $x\equiv_M y$ if and only if for every word $z$ one has $xz\in \F_{w^*} \Leftrightarrow yz\in \F_{w^*} $. Now, it is easy to see that the factors $b^i$ and $ab^i$, $0\leq i<n-1$, are each in a distinct class, while $b^{n-1}$ and $ab^{n-1}$ are together in another class, whence there are at least $2n-1$ states in the minimal automaton recognizing $\F_{w^*}$.

In view of these considerations, we define \emph{the factor automaton of a circular word} $[w]$ as the minimal automaton  recognizing  the language $\F_{w^*}$.

\section{Circular Fibonacci Words and Minimal Forbidden Factors}

In this section, we illustrate the combinatorial  results discussed in the previous section in the special case of the circular  Fibonacci words. The Fibonacci words are a paradigmatic example that often represents the limit case for some property. For example, it is well known the worst-case running time of some pattern matching algorithms is realized by the Fibonacci words (see, e.g., \cite{KMP}). 

We fix the alphabet $A=\{a,b\}$. The sequence $(f_n)_{n\geq 1}$ of Fibonacci words is defined recursively by: $f_1=b$, $f_2=a$ and $f_n=f_{n-1}f_{n-2}$ for $n>2$. The length of the word $f_n$ is the Fibonacci number $F_n$. The limit of this sequence is the infinite Fibonacci word $f=\lim_{n\to \infty}f_n=abaababaabaab\cdots$.

\begin{table}[tb]
\centering  
\begin{minipage}{.5\linewidth}
\begin{equation*}
\begin{split}
  f_{1}  &= b \\
  f_{2}  &= a \\
  f_{3}  &= ab \\
  f_{4}  &= aba \\
  f_{5}  &= abaab \\
  f_{6}  &= abaababa\\
  f_{7}  &= abaababaabaab\\
  f_{8}  &= abaababaabaababaababa \\
 \end{split}
 \end{equation*}
 \end{minipage}
 \begin{minipage}{.4\linewidth}
\begin{equation*}
\begin{split}
   \\
   \\
  u_{3}  &= \epsilon \\
  u_{4}  &= a \\
  u_{5}  &= aba \\
  u_{6}  &= abaaba\\
  u_{7}  &= abaababaaba\\
  u_{8}  &= abaababaabaababaaba \\
 \end{split}
 \end{equation*}
 \end{minipage}
 \vspace{4mm}
\caption{\label{tab:Fibowords}The first few Fibonacci words $f_n$ and the first few words $u_n$.}
\end{table}

Every Fibonacci word $f_n$ has a factor automaton with $|f_n|+1$ states (thus attaining the lower bound on the number of states that the factor automaton of a word  can have) and the structure of the factor automaton of the Fibonacci words allows one to derive several combinatorial properties of these words (cf.~\cite{Rytter,Fi11}).

We will now describe the structure of the sets $\mina_{[f_n]}$ of minimal forbidden factors of the circular Fibonacci words $[f_n]$. 

Let us recall some well-known properties of the Fibonacci words (the reader may also see \cite[Chap.~2]{Lot01}). For every $n\geq 3$, one can write $f_n=u_nab$ if $n$ is odd or $f_n=u_nba$ if $n$ is even, where $u_n$ is a palindrome. Moreover, since $f_n=f_{n-1}f_{n-2}$ and the words $u_n$ are palindromes, one has that for every $n\geq 5$ 

\begin{equation}\label{eq:central}
 f_n=u_nxy=u_{n-1}yxu_{n-2}xy=u_{n-2}xyu_{n-1}xy
\end{equation}
for letters $x,y$ such that $\{x,y\}=\{a,b\}$. Indeed, since $u_n$ is a palindrome, one has that $u_n=u_{n-1}yxu_{n-2}$  is equal to its mirror image, $\wt{u}_{n-2}xy\wt{u}_{n-1}$, which is equal to $u_{n-2}xyu_{n-1}$ since $u_{n-1}$ and $u_{n-2}$ are palindromes.
The first few Fibonacci words $f_n$ and the first few words $u_n$ are shown in Table~\ref{tab:Fibowords}.

The words $f_n$ (as well as the words $f_nf_n$) are \emph{balanced}, that is, for every pair of factors $u$ and $v$ of the same length, one has $||u|_a-|v_a||\leq 1$ (and therefore also $||u|_b-|v_b||\leq 1$).

A \emph{bispecial factor} of a word $w$ over the alphabet $A=\{a,b\}$ is a word $v$ such that $av,bv,va,vb$ are all factors of $w$. 

\begin{proposition}\label{prop:fnfn}
For every $n\geq 3$, the set of bispecial factors of the word $f_nf_n$ is $\{u_3,u_4,\ldots,u_{n}\}$. 
\end{proposition}

\begin{proof}
As it is well known (cf.~\cite[Proposition~10]{DM94}), the bispecial factors of the infinite Fibonacci word $f$ are the central words $u_n$, $n\geq 3$. Since $f_nf_n$ is a factor of $f$ (the prefix $f_{n+3}$ of $f$ can be written as $f_{n+3}=f_{n+2}f_{n+1}=f_{n+1}f_nf_nf_{n-1}$) we have that the set of bispecial factors of the word $f_nf_n$ is contained in $\bigcup_{n\geq 3}u_n$. Since by construction $u_{n-1}$ appears in $u_n$ both as a prefix and as a suffix, we are left to prove that $u_n$ is a bispecial factor of $f_nf_n$, since this will imply that also the words $u_m$, $3\leq m\leq n$, are bispecial factors of $f_nf_n$.

The claim can be easily checked for $n=3,4$. Let us then suppose $n\geq 5$. We know that $f_nf_n=u_nxyu_nxy$ for letters $x,y$ such that $\{x,y\}=\{a,b\}$. Hence, $u_nx$ and $yu_n$ appear as factors in $f_nf_n$. We will now show that also $xu_n$ and $u_ny$ appear as factors in $f_nf_n$. Indeed, using (\ref{eq:central}), we have
 \begin{align*}\label{eq:proof}
 f_nf_n&=u_nxy \cdot u_nxy\\
 &=u_{n-1}yxu_{n-2} \cdot xy \cdot  u_{n-1}yxu_{n-2}\cdot xy\\
 &=u_{n-1}y \cdot xu_{n-2} xy  u_{n-1}y \cdot xu_{n-2}xy\\
 &= u_{n-1}y \cdot xu_ny \cdot xu_{n-2}xy
\end{align*}
that gives the desired occurrences of $xu_n$ and $u_ny$.
\end{proof}

Let us now define the sequence of words $(\hat{f}_n)_{n\geq 3}$ by $\hat{f}_n=au_na$ if $n$ is odd,  $\hat{f}_n=bu_nb$ if $n$ is even. These words are known as \emph{singular words}. Analogously, we can define the sequence of words $(\hat{g}_n)_{\geq 3}$ by $\hat{g}_n=bu_nb$ if $n$ is odd,  $\hat{g}_n=au_na$ if $n$ is even. For every $n$, the word $\hat{g}_n$ is obtained from the word $\hat{f}_n$ by changing the first and the last letter. The elements of the sequence $\hat{g}_n$ are indeed the minimal forbidden factors of the infinite Fibonacci word $f$(see~\cite{Mignosi02}).
The first few values of the sequences $\hat{f}_n$ and $\hat{g}_n$ are shown in Table \ref{tab:mff}.

\begin{table}[tb]
\centering  
\begin{minipage}{.4\linewidth}
\begin{equation*}
\begin{split}
  \hat{f}_{3}  &= aa \\
  \hat{f}_{4}  &= bab \\
  \hat{f}_{5}  &= aabaa \\
  \hat{f}_{6}  &= babaabab\\
  \hat{f}_{7}  &= aabaababaabaa\\
  \hat{f}_{8}  &= babaababaabaababaabab \\
 \end{split}
 \end{equation*}
 \end{minipage}
 \begin{minipage}{.4\linewidth}
 \begin{equation*}
\begin{split}
  \hat{g}_{3}  &= bb \\
  \hat{g}_{4}  &= aaa \\
  \hat{g}_{5}  &= babab \\
  \hat{g}_{6}  &= aabaabaa\\
  \hat{g}_{7}  &= babaababaabab\\
  \hat{g}_{8}  &= aabaababaabaababaabaa \\
 \end{split}
 \end{equation*}
  \end{minipage}
 \vspace{4mm}
\caption{\label{tab:mff}The first few elements of the sequences $\hat{f}_n$ and $\hat{g}_n$.}
\end{table}

The structure of the sets of minimal forbidden factors of circular Fibonacci words can be described in terms of the words $\hat{f}_{n}$ and $\hat{g}_{n}$, as follows.
The first few sets $\mina_{[f_n]}$ are displayed in Table~\ref{tab:Fibmff}. We have $\mina_{[f_1]}=\mina_{[b]}=\{a\}$, $\mina_{[f_2]}=\mina_{[a]}=\{b\}$ and $\mina_{[f_3]}=\mina_{[ab]}=\{aa,bb\}$. The following theorem gives a characterization of the sets $\mina_{[f_n]}$ for $n\geq 4$.

\begin{table}[tb]
\centering  
\begin{equation*}
\begin{split}
  \mina_{[f_1]}  = & \{a\} \\
  \mina_{[f_2]}  = & \{b\} \\
  \mina_{[f_3]}  = & \{aa,bb\} \\
  \mina_{[f_4]}  = & \{bb,aaa,bab\} \\
  \mina_{[f_5]}  = & \{bb, aaa, aabaa, babab\} \\
  \mina_{[f_6]}  = & \{bb, aaa, babab, aabaabaa, babaabab\} \\
  \mina_{[f_7]}  = & \{bb, aaa, babab, aabaabaa, aabaababaabaa, babaababaabab\} \\
  \mina_{[f_8]}  = & \{bb, aaa, babab, aabaabaa, babaababaabab, \\
  & \ aabaababaabaababaabaa, babaababaabaababaabab\} \\
 \end{split}
 \end{equation*}
 \vspace{4mm}
\caption{\label{tab:Fibmff}The first few sets of minimal forbidden factors of the circular Fibonacci words.}
\end{table}

\begin{theorem}\label{thm:fib}
 For every $n\geq 4$, $\mina_{[f_n]}=\{\hat{g}_{3},\hat{g}_{4},\ldots,\hat{g}_{n},\hat{f}_{n}\}$.
\end{theorem}

\begin{proof}
By Lemma \ref{lem:circ},  $\mina_{[f_n]}= \mina_{f_nf_n}\cap A^{\leq |f_n|}$. Let $xuy$, $u\in A^*$, $x,y\in A$, be in $\mina_{f_nf_n}\cap A^{\leq |f_n|}$. Then, $xu$ has an occurrence in $f_nf_n$ followed by letter $\bar{y}$, the letter different from $y$, and $uy$  has an occurrence in $f_nf_n$ preceded by letter $\bar{x}$, the letter different from $x$. Therefore, $u$ is a bispecial factor of the word $f_nf_n$, hence, by Proposition~\ref{prop:fnfn}, $u\in\{u_3,u_4,\ldots,u_n\}$. Thus, an element in $\mina_{[f_n]}$ is of the form $\alpha u_i\beta$ for some $3\leq i\leq n$ and $\alpha,\beta\in A$.

We first prove that for every $3\leq i\leq n$, if $\alpha\neq \beta$, then the words $\alpha u_i\beta$ and $\beta u_i\alpha$ occur as factors in $f_nf_n$. Let us write $f_nf_n=u_nxyu_nxy$, for letters $x,y$ such that $\{x,y\}=\{a,b\}$. 

As observed in the proof of Proposition~\ref{prop:fnfn}, $f_nf_n$ contains an internal occurrence of $xu_ny=xu_{n-1}yxu_{n-2}y$, hence it contains an occurrence of $xu_{n-1}y$ and an occurrence of $xu_{n-2}y$. On the other hand, $f_nf_n=u_nxyu_nxy=u_nxy u_{n-2}xyu_{n-1} xy$, hence $f_nf_n$ contains an occurrence of $yu_nx$, an occurrence of $yu_{n-1}x$ and an occurrence of $yu_{n-2}x$. Continuing this way, recursively, we see that $f_nf_n$ contains an occurrence of $xu_iy$ and one of $yu_ix$ for every $3\leq i\leq n$.

Thus,  a word in $\mina_{[f_n]}$ can only be of the form $\hat{f}_i$ or $\hat{g}_i$, for  some $3\leq i\leq n$.

\textit{Claim: $\hat{f}_n$ is a minimal forbidden factor of $f_{n}f_{n}$.} 

Proof: Let $\hat{f}_n=xu_nx$, $x\in A$. By Proposition~\ref{prop:fnfn}, we know that $xu_n$ and $u_nx$ are factors of $f_nf_n$. It remains to show that  $\hat{f}_n$ does not occur in $f_nf_n=u_nxyu_nxy$, $y\neq x$. If $xu_nx$ occurs in $f_nf_n$, then it occurs in $u_nxyu_n$. But it is known that the longest repeated prefix of $u_nxyu_n$ is $u_n$ (cf.~\cite{cacca}), so $u_nx$ cannot appear in $u_nxyu_n$.

\textit{Claim: $\hat{f}_n$ is a factor of $f_{n+1}f_{n+1}$.} 

Proof: The first letter of $\hat{f}_n$ is equal to the last letter of $f_{n+1}$ and, by removing the first letter from $\hat{f}_n$,  one obtains a prefix of $f_{n+1}$. Hence, $\hat{f}_n$ is a factor of  $f_{n+1}f_{n+1}$. 

\textit{Claim: For every $3\leq i< n$,  $\hat{g}_i$ is a minimal forbidden factor of $f_nf_n$.} 

Proof: By the previous claim, it follows that for every $3\leq i< n$, the word $\hat{f}_i$ is factor of the word $f_nf_n$. Therefore $\hat{g}_i$ cannot be a factor of $f_nf_n$ otherwise  the word $f_nf_n$ would not be balanced. Since removing the first or the last letter from the word  $\hat{g}_i$ one obtains a factor of the word $f_nf_n$, the claim is proved.

\textit{Claim: $\hat{g}_n$ is a minimal forbidden factor of $f_nf_n$.} 

Proof: $\hat{g}_n=yu_ny$ is a minimal forbidden factor of the infinite Fibonacci word $f$, so it cannot be a factor of $f_nf_n$. On the other hand, we proved in Proposition~\ref{prop:fnfn} that $u_n$ is a bispecial factor of $f_nf_n$, hence $yu_n$ and $u_ny$ occur in $f_nf_n$.
\end{proof}

Notice that, by Lemma~\ref{lem:m}, for any circular word $[w]$, one has that $|w|$ is an upper bound on the length of the minimal forbidden factors of $[w]$. The previous theorem shows that this bound is indeed tight. However, the maximum length of a minimal forbidden factor of a circular word $[w]$ is not always equal to $|w|$. For example, for $w=aabbab$ one has $\mina_{[w]}=\{aaa, bbb, aaba, abab, babb, bbaa\}$.

About the cardinality of the set of minimal forbidden factors, however, Fibonacci words are not extremal. Indeed, by Lemma \ref{lem:nb}, we know that a binary word can  have a number of distinct minimal forbidden factors equal to its length, while in the case of the Fibonacci words this number is only logarithmic in the length of the word, as we have, by Theorem \ref{thm:fib}, that the cardinality of $\mina_{[f_n]}$ is  $n-1$ and the length of $f_n$ is exponential in $n$.

By Theorem~\ref{thm:min2}, if $\mathcal{T}$ is the trie recognizing the set $\{\hat{g}_{3},\hat{g}_{4},\ldots,\hat{g}_{n},\hat{f}_{n}\}$, then algorithm {\sc L-automaton} on the input trie $\mathcal{T}$ builds the minimal deterministic automaton recognizing $\F_{f_n^*}$. 
Since the automaton output by algorithm {\sc L-automaton} has the same set of states of the  input trie $\mathcal{T}$ after removing sink states, and since removing the last letter from each word $\hat{g}_i$ results in a prefix of $\hat{f}_{i+1}$, we have that the factor automaton of the circular Fibonacci word $[f_n]$ (that is, the minimal automaton recognizing $\F_{f_n^*}$) has exactly $2F_n-1$ states (see Figure~\ref{fig:mff} for an example). In view of Theorem \ref{thm:states}, circular Fibonacci  words have a factor automaton with the largest possible number of states, while in the linear case it is well known that Fibonacci words have a factor automaton with the smallest possible number of states. 

\begin{figure}[tb]
\begin{center}\small
\includegraphics[height=55mm]{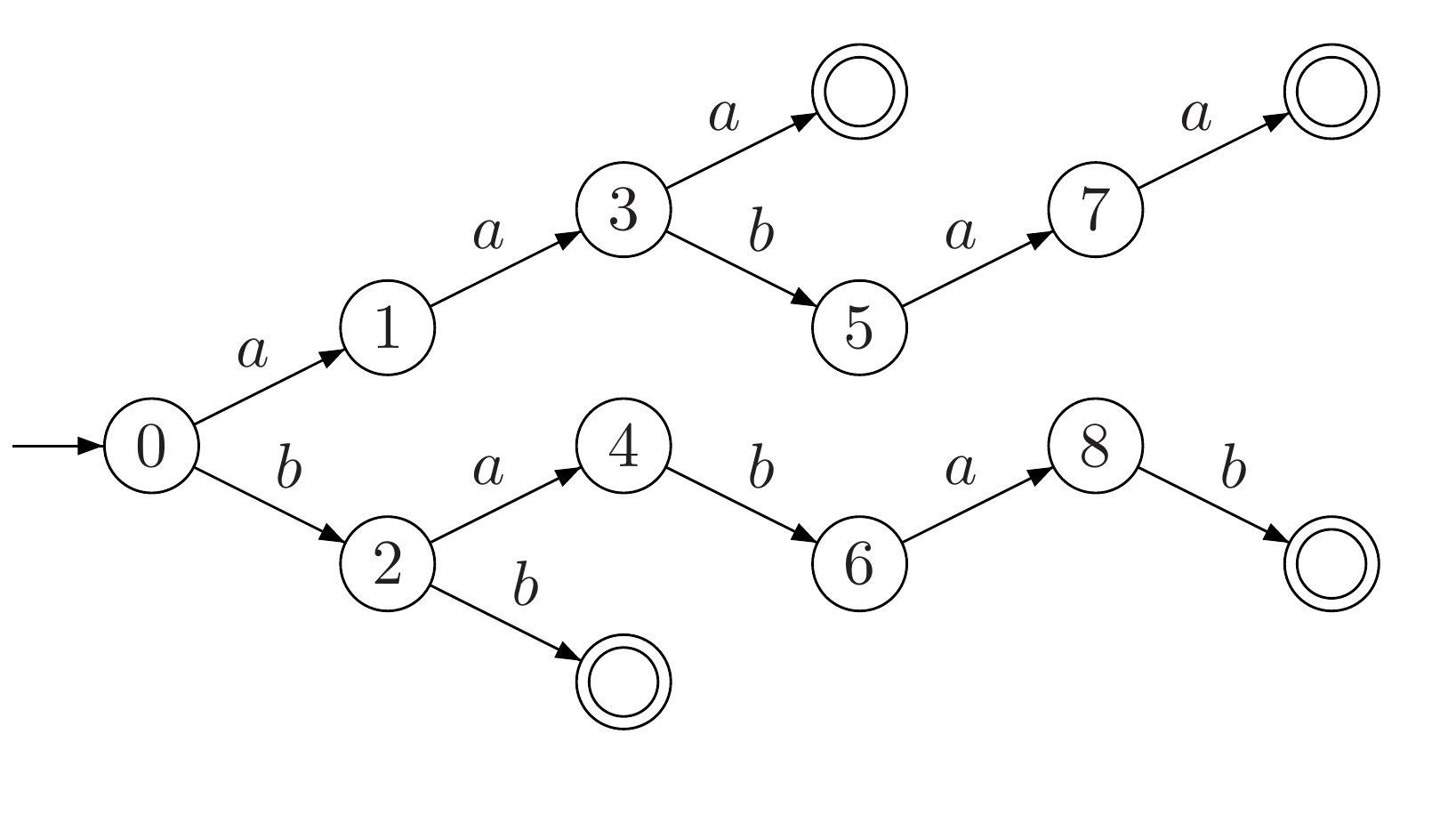}
\includegraphics[height=32mm]{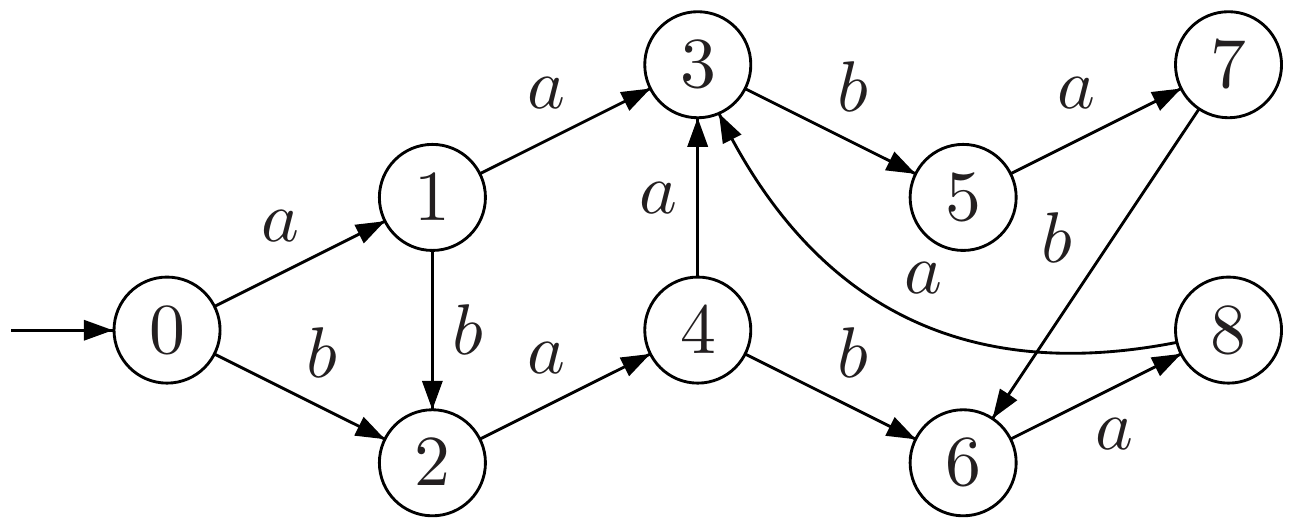}
\end{center}
\caption{The trie $\mathcal{T}$ recognizing the set $\mina_{[f_5]}$ (top), and the automaton output by algorithm  {\sc L-automaton} on the input trie $\mathcal{T}$ after removing sink states (bottom), which is the minimal automaton recognizing $\F_{f_5^*}$. It has $9=2F_5-1$ states.\label{fig:mff}}
\end{figure}

\section{Conclusions and Open Problems}

We investigated combinatorial properties of minimal forbidden factors of circular words.

We proved that the automaton built by  algorithm {\sc L-automaton} on the input trie recognizing the set of minimal forbidden factors of a circular word is minimal. More generally, it would be interesting to characterize those antifactorial languages for which algorithm {\sc L-automaton} builds a minimal automaton.

In addition to being interesting from the point of view of formal language theory, we believe the study of minimal forbidden factors of circular words will also lead to new applications in sequence analysis. An example on this direction is given in \cite{InfComp18}.

\section{Acknowledgements}

We thank anonymous referees for several valuable comments.

\bibliographystyle{abbrv}
\bibliography{ref}

\end{document}